%% file: main.tex
\documentclass[11pt]{article}

\usepackage[preprint]{acl}

\usepackage{times}
\usepackage{latexsym}
\usepackage[T1]{fontenc}
\usepackage[utf8]{inputenc}
\usepackage{microtype}
\usepackage{inconsolata}

\usepackage{graphicx}
\usepackage{subcaption}
\usepackage{booktabs}
\usepackage{amsmath}
\usepackage{amssymb}
\usepackage{mathtools}
\usepackage{amsthm}
\usepackage{multirow}
\usepackage{xcolor}
\usepackage{colortbl}
\usepackage{changepage}
\usepackage{threeparttable}
\usepackage{pifont}
\usepackage{enumitem}
\usepackage[normalem]{ulem}
\useunder{\uline}{\ul}{}
\usepackage[capitalize,noabbrev]{cleveref}
\usepackage[most]{tcolorbox}
\usepackage{algorithm}
\usepackage{algorithmic}
\usepackage[disable,textsize=tiny]{todonotes}

\newcommand{\method}{BOSQ}
\newcommand{\inlinecomment}[1]{\textcolor{gray}{\small\sffamily\(\triangleright\) #1}}

\newtcolorbox{promptbox}[1]{
    colback=gray!5!white,
    colframe=gray!75!black,
    fonttitle=\bfseries,
    title=#1,
    arc=2mm,
    enhanced,
    boxrule=0.5pt,
    left=1mm,
    right=1mm,
    top=0.75mm,
    bottom=0.75mm,
    before skip=0.5em,
    after skip=0.7em
}

\theoremstyle{plain}
\newtheorem{theorem}{Theorem}[section]

\theoremstyle{definition}

\theoremstyle{remark}

\definecolor{myred}{rgb}{0.75, 0.45, 0.50}
\definecolor{myblue}{rgb}{0.15, 0.40, 0.70}
\definecolor{myorange}{rgb}{0.85, 0.50, 0.15}
\definecolor{mygray}{rgb}{0.5, 0.5, 0.5}
\definecolor{mypurple}{rgb}{0.40, 0.35, 0.60}
\definecolor{mygreen}{rgb}{0.30, 0.65, 0.55}
\definecolor{mybrown}{rgb}{0.60, 0.35, 0.30}

\title{Scaling GraphLLM with Bilevel-Optimized Sparse Querying}

\author{
 \textbf{Yangzhe Peng\textsuperscript{1}},
 \textbf{Haiquan Qiu\textsuperscript{2}},
 \textbf{Quanming Yao\textsuperscript{2}},
 \textbf{Kun He\textsuperscript{1}},
\\
\\
 \textsuperscript{1}Huazhong University of Science and Technology,
 \textsuperscript{2}Tsinghua University,
}

\begin{document}
\maketitle

\begin{abstract}
LLMs have recently shown strong potential in enhancing node-level tasks on text-attributed graphs (TAGs) by providing explanation features.
However, their practical use is severely limited by the high computational and monetary cost of repeated LLM queries.
To illustrate, naively generating explanations for all nodes on a medium-sized benchmark like Photo (48k nodes) using a representative method (e.g., TAPE) would consume \emph{days} of processing time.
In this paper, we propose Bilevel-Optimized Sparse Querying (\method{}), a general framework that selectively leverages LLM-derived explanation features to enhance performance on node-level tasks on TAGs.
We design an adaptive sparse querying strategy that selectively decides when to invoke LLMs, avoiding redundant or low-gain queries and significantly reducing computation overhead.
Extensive experiments on six real-world TAG datasets involving two types of node-level tasks demonstrate that \method{} runs substantially faster than existing GraphLLM methods while consistently delivering on-par or superior performance.
Our code is available at \url{https://anonymous.4open.science/r/BOSQ-0336}.
\end{abstract}

\input{sections/01-intro}
\input{sections/03-method}

\input{sections/04-exp}
\input{sections/02-rw}
\input{sections/05-conclusion}

\section*{Limitations}

While \method{} establishes a principled bilevel optimization framework for sparse graph querying, several limitations remain. First, our current implementation focuses on the generative LLM-as-Explainer paradigm; extending the selector to discriminative LLM-as-Encoder settings may further broaden its applicability. Second, the method uses a fixed cardinality budget, and future work may consider dynamic budgeting or nucleus-style sparse selection to support finer-grained graph-wise adaptivity. Finally, the current selector uses an efficient learnable vector parameterization. More expressive non-linear selectors, such as MLP-based architectures, could capture richer topological dependencies, but may require additional tuning and introduce extra computational overhead.

\section*{Ethical Considerations}

This work aims to reduce the computational and monetary cost of LLM-enhanced graph learning by avoiding unnecessary LLM queries. The method does not introduce new data collection, human-subject experiments, or model deployment decisions. As with other methods that use LLM-generated text features, downstream users should consider potential biases, privacy issues, or licensing constraints in the text-attributed graphs and LLMs they use.

\bibliography{ref}

\appendix
\input{sections/06-appendix}

\end{document}

%% file: sections/01-intro.tex
\section{Introduction}

Recently, Large Language Models (LLMs) have demonstrated strong semantic comprehension and reasoning capabilities~\cite{touvron2023}, leading to significant interest in exploring their utility for Text-Attributed Graphs (TAGs), where nodes are characterized by rich textual information~\cite{chen2024c,fang2025,hu2020b}. However, the direct integration of LLMs and graphs remains challenging: while LLMs excel at sequential reasoning, they struggle to inherently capture the non-Euclidean structural dependencies and relational inductive biases that characterize graph-structured data.

\begin{figure}[t]
  \begin{center}
    \centerline{\includegraphics[width=\columnwidth]{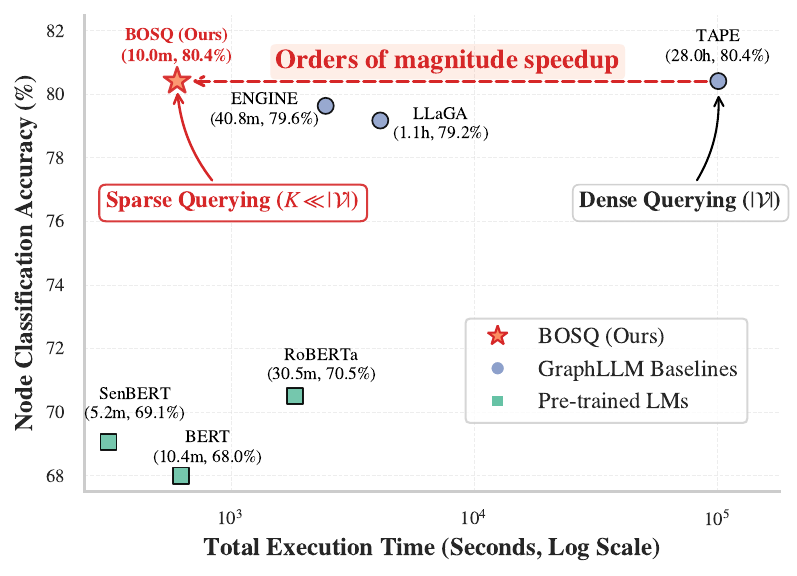}}
    \vspace{-2mm}
    \caption{
        \method{} improves the efficiency-performance trade-off of existing GraphLLM frameworks. By learning which nodes most benefit from LLM-generated explanations under a sparse query budget, \method{} achieves substantially faster end-to-end execution than TAPE while preserving competitive task performance. Our method (red star) sits near the ideal top-left region, offering a practical and scalable solution for large-scale TAGs where dense GraphLLM baselines are infeasible. Please refer to Table~\ref{tab:clf-exp} for detailed results.
        }
        \label{fig:acc_vs_time}
    \end{center}
    \vspace{-6mm}
\end{figure}

To bridge this gap, recent studies have proposed integrating LLMs with graph learning methods, forming several emerging paradigms collectively called GraphLLMs~\cite{li2024e}. Specifically, three main approaches have been identified~\cite{wu2025a}: \text{1.} The \textbf{Encoder} paradigm utilizes LLMs to encode nodes' textual attributes, generating more expressive feature representations ($\mathbf{h}^{\text{LLM}}$) that surpass shallow embeddings, which are then aggregated by downstream GNNs~\cite{zhu2024b}. \text{2.} The \textbf{Predictor} paradigm employs LLMs as direct classifiers or regressors, consuming structured prompts ($\mathcal{P}$) constructed from nodes' text and their aggregated neighborhood contexts to produce final, text-based predictions~\cite{chen2024c,tang2024b}. The process can be abstracted as: $\text{Pred}_i = \text{LLM}(\mathcal{P}_i)$. \text{3.} The \textbf{Explainer} paradigm leverages LLMs' generative capability to enhance node attributes by generating a detailed, semantic-rich \textbf{explanation feature} $\mathbf{e}_n$ \cite{he2024}. Following this terminology, ``explainer'' refers to LLM-based explanation-feature generation rather than post-hoc model interpretability. This feature is derived from the node's textual context, and is designed to augment the node's original information, thereby enriching the node features before being processed by GNNs.

However, a central obstacle to making LLM-enhanced TAG methods broadly practical is the cost of repeatedly invoking large models over graph-structured instances. While LLMs can provide rich semantic augmentation, such augmentation is typically much more expensive than GNN message passing and often has to be applied at the node or instance level. As the coverage of LLM-augmented nodes increases, the computational and monetary cost can therefore grow rapidly. In previous explainer-style pipelines, this manifests as generating explanation features for all nodes before downstream graph learning:
{\setlength{\abovedisplayskip}{0.5em}
\setlength{\belowdisplayskip}{0.5em}
\[
C_{\text{total}} \propto |\mathcal{V}| \times C_{\text{LLM}}
\]
}
where $C_{\text{LLM}}$ denotes the cost of a single LLM query. For graphs of realistic size, e.g., tens or hundreds of thousands of nodes, this translates into days or even weeks of processing time and significant financial expense, rendering the approach infeasible for practical deployment, especially on medium to large-scale benchmark datasets.

Moreover, the assumption that every node contributes equally to the downstream task contradicts the principle of information sparsity observed in both biological and artificial neural networks~\cite{bengio2017consciousness,gao2025weight}. Just as human cognition allocates attention selectively (the ``spotlight'' metaphor)~\cite{posner1980orienting} and recent sparse LLM architectures~\cite{lu2025,yuan2025} dynamically route computation to a subset of relevant components, graph data inherently exhibits structural and semantic redundancy. This observation suggests that the ``brute-force'' approach of dense querying is not only inefficient but potentially suboptimal due to noise accumulation. Inspired by these mechanisms of selective attention,
we propose \textbf{Bilevel-Optimized Sparse Querying (\method{})}, an efficient and general explanation-guided framework. The core of \method{} is an \textbf{adaptive sparse querying strategy} that intelligently selects a small, fixed-size budget of $K$ nodes ($K \ll |\mathcal{V}|$) for LLM interaction as illustrated in Figure~\ref{fig:motivation}. This selection is guided by learnable scores assigned to each node, which are optimized directly by backpropagating gradients from the downstream task's validation loss. This task-driven mechanism ensures that the limited query budget is allocated to nodes whose explanation features provide the highest utility for improving the model's performance, striking an optimal balance between efficiency and effectiveness.

Our contributions are summarized as follows:
\begin{itemize}[nosep]
    \item We propose \textbf{\method{}}, a novel framework that efficiently integrates LLM-derived explanation features to boost node-level task performance on TAGs.
    \item \textbf{\method{}} employs an \textbf{adaptive sparse querying strategy} leveraging bilevel optimization to dynamically prioritize and select nodes for LLM explanation queries.
    \item We empirically demonstrate that \method{} runs substantially faster on six real-world TAG datasets while consistently maintaining competitive or even superior predictive performance, yielding a stronger overall efficiency-performance trade-off.
\end{itemize}

\begin{figure*}[t!]
  \begin{center}
    \centerline{\includegraphics[width=2\columnwidth]{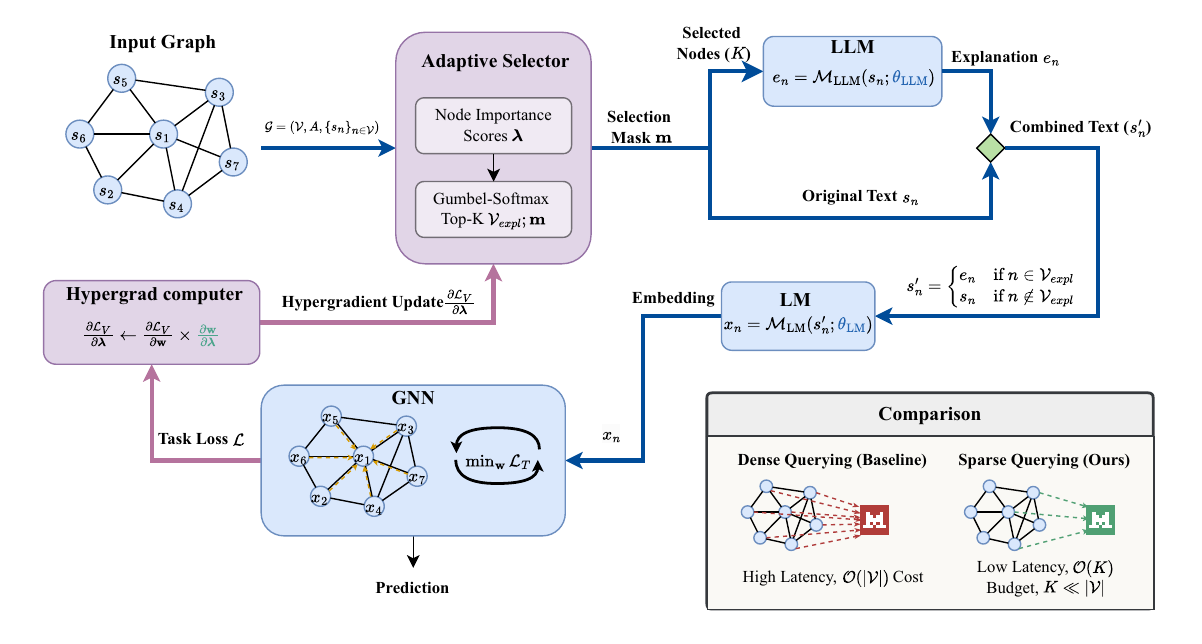}}
    \vspace{-3mm}
    \caption{
        Overview of the \method{} framework. Our method treats node selection as a bilevel optimization problem. The outer loop (purple) learns an Adaptive Selector to identify a sparse subset of nodes that benefit most from LLM explanations, while the inner loop (blue) optimizes a GNN on the resulting augmented graph. By using hypergradients to guide the selection mask $\mathbf{m}$, \method{} selectively invokes the LLM only for critical nodes ($K \ll |\mathcal{V}|$). This mechanism achieves a superior trade-off between task performance and computational efficiency compared to previous dense querying baselines (see Comparison box).
        }
        \label{fig:motivation}
    \end{center}
    \vspace{-8mm}
\end{figure*}

%% file: sections/03-method.tex
\section{Method}

In this section, we present the detailed methodology of Bilevel-Optimized Sparse Querying (\method{}), a framework that efficiently incorporates Large Language Model (LLM)-derived explanation features to enhance node-level tasks on Text-Attributed Graphs (TAGs). The core idea lies in an adaptive sparse querying strategy that selects a subset of nodes for explanation generation to balance performance and computational cost.

\subsection{Preliminary}

\textbf{Text-attributed graphs.} Formally, a text-attributed graph (TAG)~\cite{ma2021} can be represented as \(\mathcal{G} = (\mathcal{V}, A, \{s_n\}_{n \in \mathcal{V}})\), where \(\mathcal{V}\) is a set of \(N\) nodes, \(A \in \mathbb{R}^{N \times N}\) is the adjacency matrix, and \(s_n \in \mathcal{D}^{L_n}\) is a sequential text associated with node \(n \in \mathcal{V}\), with \(\mathcal{D}\) as the words or tokens dictionary, and \(L_n\) as the sequence length.

\textbf{Node level tasks.} Consider a graph \(\mathcal{G}\) with a set of nodes \(\mathcal{V}\). The nodes are partitioned into a labeled set \(\mathcal{V}_l \subset \mathcal{V}\) and an unlabeled set \(\mathcal{V}_u\). The general objective of both node classification (node clf) and node regression (node reg) is to train a graph-based neural network model, utilizing the structure of \(\mathcal{G}\) and the known values of \(\mathcal{V}_l\) (which can be discrete class labels or continuous values), to predict the corresponding unknown values for the nodes in \(\mathcal{V}_u\). While the fundamental goal is shared, the definition of the unlabeled set \(\mathcal{V}_u\) differs: for node classification, \(\mathcal{V}_u\) comprises all remaining nodes (\(\mathcal{V}_u = \mathcal{V} \setminus \mathcal{V}_{l}\)), whereas for node regression, \(\mathcal{V}_u\) only represents a specified subset of nodes (e.g., a particular type of node in a heterogeneous graph) whose continuous value is targeted for prediction.

 
\textbf{LLM-as-Explainer Paradigm} Recently, TAPE~\cite{he2024} have demonstrated that instead of directly encoding $X_v$ using a shallow Language Model (LM), invoking a Large Language Model (LLM) to generate explanatory features significantly boosts performance. Formally,
for each node $n$, a Large Language Model (LLM), denoted as $\mathcal{M}_{\text{LLM}}$ with parameters ${\theta_{\text{LLM}}}$, is queried to transform the raw text $s_n$ into a more informative explanatory text $e_n$:
\begin{equation}
e_n = \mathcal{M}_{\text{LLM}}(s_n; {\theta_{\text{LLM}}})
\end{equation}

The original text $s_n$ and the LLM-generated explanation $e_n$ are synthesized into an enriched textual attribute $s'_n = \mathcal{T}(s_n, e_n)$. $\mathcal{T}$ can be implemented as a simple concatenation (i.e., $s_n \, \Vert \, e_n$) or a complete replacement (i.e., $s'_n = e_n$ or $s'_n = s_n$). 

The integrated text \(s'_n\) is converted into a \(D\)-dimensional node embedding \(x_n\) using a Language Model (LM, e.g. RoBERTa), \(\mathcal{M}_{\text{LM}}\) with parameters \({\theta_{\text{LM}}}\):
\begin{equation}
x_n = \mathcal{M}_{\text{LM}}(s'_n; {\theta_{\text{LM}}})
\end{equation}

A Graph Neural Network (GNN), \(\mathcal{M}_{\text{GNN}}\) with parameters \({\mathbf{w}}\), takes as input the feature matrix \(X = [x_1, \dots, x_N]^T \in \mathbb{R}^{N \times D}\) and the adjacency matrix \(A\) to produce predictions \(\hat{Y}\). The GNN is trained by minimizing the loss function \(\mathcal{L}\) over the training set \(\mathcal{V}_{\text{train}}\). The prediction and optimization objectives are defined as follows:
\begin{gather}
\hat{Y} = \mathcal{M}_{\text{GNN}}(X, A; {\mathbf{w}})\\[-0.5em]
\min_{{\mathbf{w}}} \mathcal{L}_T(\hat{Y}, Y) = \min_{{\mathbf{w}}} \sum_{n \in \mathcal{V}_{\text{train}}} \ell(\hat{y}_n, y_n)
\end{gather}
\subsection{Bilevel-Optimized Sparse Querying}

The core efficiency of our framework stems from the selective application of LLM explanations to a subset of nodes. To achieve this, we define a selection mask $\mathbf{m} \in \{0, 1\}^N$, where $m_n = 1$ indicates that node $n$ is selected for LLM explanation, and $m_n = 0$ otherwise.

However, directly optimizing $\mathbf{m}$ is a non-trivial combinatorial problem, which prevents the use of gradient-based optimization due to its discrete nature. While one could rely on static heuristics (e.g., selecting nodes based on degree or entropy), such methods are task-agnostic and may fail to capture the dynamic contribution of specific nodes to the downstream performance.

To bridge this gap, we propose to relax the discrete mask into a continuous importance vector $\boldsymbol{\lambda} = (\lambda_1, \dots, \lambda_N)$. By treating $\boldsymbol{\lambda}$ as a differentiable proxy for the selection mask $\mathbf{m}$, we can formulate node selection as a bilevel optimization problem. In this setup, the outer loop optimizes the node importance scores $\boldsymbol{\lambda}$ to maximize generalization on validation data, while the inner loop updates the model parameters $\mathbf{w}$ based on the selected information. The transformation from continuous scores $\boldsymbol{\lambda}$ to the final discrete mask $\mathbf{m}$ is then performed through a top-$K$ sparsification process, which we will detail in Section~\ref{sec:detail method}.

Formally, the bilevel optimization problem can be formulated as:
\vspace{-0.5mm}
\begin{equation}
\begin{aligned}\label{eq: bilevel opt}
\boldsymbol{\lambda}^*
&\coloneqq
\underset{\boldsymbol{\lambda}}{\operatorname{argmin}}\,
\mathcal{L}_{\text{V}}^*(\boldsymbol{\lambda}), \\
\mathcal{L}_{\text{V}}^*(\boldsymbol{\lambda})
&\coloneqq
\mathcal{L}_{\text{V}}\big(\boldsymbol{\lambda},
\mathbf{w}^*(\boldsymbol{\lambda}) \big), \\
\mathbf{w}^*(\boldsymbol{\lambda})
&\coloneqq
\underset{\mathbf{w}}{\operatorname{argmin}}\,
\mathcal{L}_{\text{T}}(\boldsymbol{\lambda}, \mathbf{w})
\end{aligned}
\end{equation}

where $\mathcal{L}_{\text{T}}$ and $\mathcal{L}_{\text{V}}$ denote the training and validation losses, respectively. By computing the hypergradient of the validation loss with respect to $\lambda_n$, we can accurately quantify each node's influence on the model's generalization ability, allowing the framework to prioritize explanations for the most ``informative'' nodes.

Equipped with the selection mask $\mathbf{m}$ (derived from $\boldsymbol{\lambda}$), we design a Sparse Querying architecture that selectively integrates LLM-derived features into the GNN pipeline. The forward pass proceeds through three core components:
\begin{equation}
  s'_n = 
  \begin{cases} 
    \mathcal{M}_{\text{LLM}}(s_n; {\color{myblue} \theta_{\text{LLM}}}) & \text{if } n \in \mathcal{V}_{expl} \\ 
    s_n & \text{if } n \notin \mathcal{V}_{expl} 
  \end{cases}
\end{equation}
\begin{equation}
  x_n = \mathcal{M}_{\text{LM}}(s'_n; {\color{myblue} \theta_{\text{LM}}})
\end{equation}
\begin{equation}
  \hat{Y} = \mathcal{M}_{\text{GNN}}(X, A; {\color{myorange} \mathbf{w}})
\end{equation}
where $\mathcal{V}_{expl} \subset \mathcal{V}$ represents the selected nodes to be explained. To ensure computational efficiency, we keep the parameters of the LLM (${\color{myblue} \theta_{\text{LLM}}}$) and the LM (${\color{myblue} \theta_{\text{LM}}}$) frozen, only optimizing the GNN parameters ${\color{myorange} \mathbf{w}}$.
The loss function $\ell$ is instantiated as cross-entropy for node classification and $L_1$ loss for node regression, following the standard practice in existing GraphLLM benchmarks \cite{robinson2024, wu2025a}. Algorithm \ref{alg:bilevel} illustrates the complete framework of our method.

\subsection{Task-driven Optimization and Implementation}\label{sec:detail method}

Having established the bilevel framework in equation~\ref{eq: bilevel opt}, we now describe the task-driven optimization process for the importance scores $\boldsymbol{\lambda}$ and the subsequent derivation of the selection mask $\mathbf{m}$. We optimize $\boldsymbol{\lambda}$ by computing the validation hypergradient with respect to the node importance scores:
\vspace{-0.5mm}
\begingroup
\setlength{\abovedisplayskip}{0.2em}
\setlength{\belowdisplayskip}{0.2em}
\small
\begin{equation}
\begin{aligned}
\frac{\partial \mathcal{L}_{\text{V}}^*(\boldsymbol{\lambda})}{\partial \boldsymbol{\lambda}}
= &\;
{\color{mygray}
\frac{\partial \mathcal{L}_{\text{V}}(\boldsymbol{\lambda}, \mathbf{w}^*)}
{\partial \boldsymbol{\lambda}}
}
&+
\frac{\partial \mathcal{L}_{\text{V}}(\boldsymbol{\lambda}, \mathbf{w}^*)}
{\partial \mathbf{w}^{*}}
{\color{mygreen}
\frac{\partial \mathbf{w}^*}{\partial \boldsymbol{\lambda}}}
\end{aligned}
\end{equation}
\vspace{-0.5em}
\endgroup

Following the standard convention in gradient-based hyperparameter optimization~\cite{bengio2000, franceschi2018}, we evaluate the optimized model on a held-out validation set and use this validation loss as the outer objective. In our formulation, the adaptive node scores $\boldsymbol{\lambda}$ affect the validation objective only through the inner-loop solution $\mathbf{w}^*(\boldsymbol{\lambda})$.

Consequently, the direct gradient term $\partial \mathcal{L}_{V} / \partial \boldsymbol{\lambda}$ vanishes identically. The remaining indirect term is the relevant channel for optimizing sparse queries, and we approximate this inner-response term with a Neumann series for tractability~\cite{liu2022}.

\begin{theorem}[Sparse Query Hypergradient]\label{thm:sparse-query-hypergradient}
Under the bilevel formulation in Equation~\ref{eq: bilevel opt} and the hypergradient decomposition above, the validation hypergradient used to update the sparse-query scores is
\begin{equation}\label{eq:sparse-hypergradient}
\begin{aligned}
\left.
\frac{\partial \mathcal{L}_{V}^{*}}
{\partial \boldsymbol{\lambda}}
\right|_{\boldsymbol{\lambda}'}
&= -
\frac{\partial \mathcal{L}_{V}}
{\partial \mathbf{w}^{*}}
{\color{mybrown}
\lim_{i \to \infty} \sum_{j=0}^i
\left( I -
\frac{\partial^2 \mathcal{L}_{T}}
{\partial \mathbf{w} \partial \mathbf{w}^T}
\right)^j} \\
&\quad \times
\frac{\partial^2 \mathcal{L}_{T}}
{\partial \mathbf{w} \partial \boldsymbol{\lambda}^T}
\Big|_{\boldsymbol{\lambda}', \mathbf{w}^*(\boldsymbol{\lambda}')}
\end{aligned}
\end{equation}
where the Neumann-series term approximates the response of the inner GNN optimum to the selector scores.
\end{theorem}

The derivation can be found in Appendix \ref{prf:ift}. Detailed steps to approximate hypergradient computation are as Algorithm \ref{alg:hyper-grad}.

Although the hypergradient updates continuous scores $\boldsymbol{\lambda}$, the actual sparse querying requires a binary mask $\mathbf{m} \in \{0,1\}^N$. A naive continuous relaxation (e.g., Gumbel-Softmax) fails to reduce costs, as the expectation-based loss definition \(\sum_n \lambda_n \cdot \ell_n(\hat y_n, y_n)\) inherently requires evaluating every node in the graph.  To achieve genuine efficiency, we employ a Straight-Through Estimator (STE) via the stop-gradient trick, the detailed procedure of which is illustrated in Algorithm \ref{alg:gumbel-softmax}. This allows the forward pass to perform a hard $K$-hot discrete selection, effectively pruning $(N-K)$ nodes and minimizing expensive operations (e.g., LLM queries). To ensure stable gradient estimation, we scale the softmax probability vector $\mathbf{p}$ by $K$, aligning its magnitude with the $K$-hot mask $\mathbf{h}$ (where $\sum h_i = K$). This ensures the framework is both computationally sparse and end-to-end differentiable~\cite{kool2019}. Formally:
\begin{equation}
\mathbf{m} = K \cdot \mathbf{p} + \operatorname{stop\_gradient}(\mathbf{h} - K \cdot \mathbf{p})
\end{equation}
where \(\mathbf{h}\) is the discrete $K$-hot mask obtained via Top-$K$ sampling, and \(\mathbf{p}\) is the softmax probability vector (where $\sum p_i = 1$). 
During the forward pass, the stop-gradient operation ensures $\mathbf{m} = \mathbf{h}$, restricting computation to selected samples only. During the backward pass, the gradient $\frac{\partial \mathbf{m}}{\partial \boldsymbol{\lambda}}$ is approximated by $\frac{\partial (K \cdot \mathbf{p})}{\partial \boldsymbol{\lambda}}$, providing a smooth and informative signal for optimization.

\begin{algorithm}[tb]
  \caption{Bilevel-Optimized Sparse Querying}
  \label{alg:bilevel}
  \begin{algorithmic}[1]
    \STATE \textbf{Input:} initial node importance scores \(\boldsymbol{\lambda}\), temperature \(\tau\), top-\(K\), outer steps \(T\), inner steps \(I\), dataset \(\mathcal{G} = (\mathcal{V}, A, \{s_n\}_{n \in \mathcal{V}})\)
    \STATE \textbf{Output:} optimized model parameters \(\mathbf{w}\) and node importance scores \(\boldsymbol{\lambda}\)

    \FOR{$t=1$ \textbf{to} $T$}
      \STATE \inlinecomment{Differentiable sparse selection}
      \STATE $\mathbf{m} \gets \texttt{GumbelTopK}(\boldsymbol{\lambda}, \tau, K)$ 
      \STATE \inlinecomment{Selective LLM querying}
      \STATE $\mathbf{X} \gets \texttt{SparseQuerying}(\{s_n\}_{n \in \mathcal{V}}, \mathbf{m})$ 
      
      \STATE \inlinecomment{Lower-level: Optimize GNN parameters}
      \STATE $\mathbf{w} \gets \texttt{Initialize}(\cdot)$
      \FOR{$i=1$ \textbf{to} $I$}
        \STATE $\mathcal{L}_{T} \gets \ell(\texttt{GNN}(\mathbf{X}, \mathbf{A}; \mathbf{w}), \mathbf{Y}_{T})$
        \STATE Update \(\mathbf{w}\) by gradient step \(\frac{\partial \mathcal{L}_{T}}{\partial \mathbf{w}}\)
      \ENDFOR

      \STATE \inlinecomment{Upper-level: Optimize selection strategy}
      \STATE $\mathcal{L}_{V} \gets \ell(\texttt{GNN}(\mathbf{X}, \mathbf{A}; \mathbf{w}), \mathbf{Y}_{V})$
      \STATE $\frac{\partial \mathcal{L}_{V}}{\partial \boldsymbol{\lambda}} \gets \texttt{HyperGrad}(\mathcal{L}_{V}, \mathcal{L}_{T}, \boldsymbol{\lambda}, \mathbf{w})$
      \STATE Update \(\boldsymbol{\lambda}\) by gradient step \(\frac{\partial \mathcal{L}_{V}}{\partial \boldsymbol{\lambda}}\)
      
      \STATE $\tau \gets \max(\tau \cdot \gamma, \tau_{\min})$ \COMMENT{Temperature annealing}
    \ENDFOR
\end{algorithmic}
\end{algorithm}

\noindent
By querying only \(K\) nodes per sparse stage, \method{} keeps LLM-generation cost independent of graph size, yielding the following complexity guarantee (details in Appendix~\ref{sec:appendix-complexity}).

\begin{theorem}[Time and space complexity of \method{}]\label{thm:complexity}
With a fixed optimization schedule and sparse query budget \(K \ll |\mathcal{V}|\), the dominant LLM-query time complexity of \method{} is
\[
\mathcal{O}(K \cdot C_{\text{LLM}}).
\]
Its simplified peak space complexity is
\[
\mathcal{O}(M_{\mathrm{GNN}} + M_{\mathrm{LLM}} + M_{\mathrm{LM}}),
\]
where \(C_{\text{LLM}}\) denotes the average per-query LLM generation cost, and \(M_{\mathrm{GNN}}\), \(M_{\mathrm{LLM}}\), and \(M_{\mathrm{LM}}\) denote the peak memory of the GNN, LLM, and LM components, respectively.
\end{theorem}

\subsection{Advantages over Alternative Approaches}


In dense explainer-style GraphLLM pipelines, the scalability bottleneck lies in the LLM-generation stage: when explanation features are generated for every node, this stage alone contributes \(O(|\mathcal{V}| \cdot C_{\text{LLM}})\) cost before downstream graph learning~\cite{he2024}. More broadly, node-level LLM augmentation in GraphLLM pipelines faces a coverage-cost trade-off, since improving coverage over graph instances typically requires more LLM-side processing~\cite{he2024,zhu2024b}.
In contrast, we introduce a \textbf{sparse selective querying mechanism that decouples the number of LLM queries from the total node count~\(N\)}. By capping invocations at a constant budget $K \ll N$, our design yields the dominant LLM-query complexity in Theorem~\ref{thm:complexity}. This substantially reduces the LLM-side overhead compared to dense node-level augmentation, enabling a more practical efficiency-performance trade-off on large graphs.
Furthermore, to maintain competitive performance, previous approach~\cite{he2024} necessitates fine-tuning the full LM parameters and operates multi-stream pipelines—maintaining separate fine-tuned LMs and GNNs for the original text and explanations, respectively. This architectural redundancy significantly escalates the per-node computational cost. Conversely, our method freezes the LM parameters and employs a unified single-stream pipeline, eliminating this redundancy to further minimize the resource consumption per node.

Another key strength of our method lies in its ability to \textbf{autonomously identify and prioritize the most informative nodes for the specific downstream task}. Unlike heuristic-based baselines that rely on predefined, static metrics to estimate node utility, our framework features a fully differentiable, end-to-end learning process for node selection. This task-driven update mechanism, guided directly by validation feedback, ensures that the limited query budget~\(K\) is dynamically allocated to nodes offering the highest marginal utility.




%% file: sections/04-exp.tex
\section{Experiment}

In this section, we conduct experiments to evaluate the proposed \method{} framework on two types of node-level tasks: node classification and node regression. Since \method{} targets practical LLM-enhanced graph learning, our main evaluation emphasizes the joint efficiency-performance trade-off. We investigate two core questions: \textbf{Q1}: How does \method{} perform overall on node classification and regression tasks? \textbf{Q2}: For large-scale graphs, how significantly does sparse querying reduce computational overhead? Additional analyses of the selection strategy, score transferability, LLM backbone robustness, and hyper-parameter sensitivity are reported in Appendix~\ref{sec: ablation-analysis}, Appendix~\ref{sec: importance-transferability-anslysis}, Appendix~\ref{sec:llm-backbone-robustness}, and Appendix~\ref{app:sensitivity_k}.

\subsection{Experimental Setup}
\textbf{Datasets.} 
We conduct the main experiments on six real-world text-attributed graph datasets covering node classification (Instagram, Photo, Computer -- homogeneous graphs from \cite{wu2025a}) and node regression (User-ltv, Item-ltv, Post-votes -- heterogeneous graphs from \cite{robinson2024}, with downsampling applied due to the large original size). Basic dataset statistics, including the number of nodes and edges and the average token count, are summarized in Table~\ref{tab:dataset-stats}. The table also includes the additional large-scale datasets used outside the main overall-performance tables: ogbn-products for the million-scale study in Section~\ref{subsec:large-scale} and arXiv for the additional scalability study in Appendix~\ref{sec: arxiv-scalability}. Detailed information and dataset splitting methods are presented in Appendix~\ref{sec: appendix-dataset}. Unless otherwise specified, all reported time metrics refer to the total end-to-end execution time, covering preprocessing, training, and testing stages.

\begin{table}[t]
\centering
\caption{Basic dataset statistics. The first six datasets correspond to the main overall-performance experiments: Photo, Instagram, and Computer are node classification tasks, while User-ltv, Item-ltv, and Post-votes are node regression tasks. User-ltv and Item-ltv represent different node types (User nodes and Book item nodes, respectively) within the same heterogeneous graph. ogbn-products is used in the million-scale study, and arXiv is reported in the additional scalability study.}
\label{tab:dataset-stats}
\small
\resizebox{\columnwidth}{!}{%
\begin{tabular}{cccc}
\hline
Dataset    & \#Nodes & \#Edges & Avg. \#Token\\ \hline
Photo      & 48,362  & 873,793 & 182.30       \\
Instagram  & 11,339  & 155,349 & 46.80        \\
Computer   & 87,229  & 1,256,548 & 111.66       \\
User-ltv   & 35,772  & 70,632  & 132.62      \\
Item-ltv   & 35,772  & 70,632  & 132.62      \\
Post-votes & 40,947  & 91,142  & 154.15      \\
ogbn-products & 2,449,029 & 123,718,152 & 145.50 \\
arXiv & 169,343 & 1,166,243 & 216.95      \\ \hline
\end{tabular}
}
\vspace{-3mm}
\end{table}

\textbf{Baselines.} To evaluate the effectiveness of our proposed \method{}, we compare it against a diverse set of baselines covering traditional GNNs, pretrained LMs, and LLM-augmented graph methods. For the node classification task, we employ twelve baselines including the classical GNN models GCN~\citep{kipf2017}, SAGE~\citep{hamilton2017}, and GAT~\citep{velickovic2018}; pretrained LM baselines SentenceBERT-66M~\cite{reimers2019}, RoBERTa-355M~\cite{liu2019a}, and BERT~\cite{devlin2019}; state-of-the-art LLM-based textual graph learning methods TAPE~\citep{he2024}, ENGINE\cite{zhu2024b}, and LLaGA\cite{chen2024c}; as well as GCN variants enhanced with LLM-derived features, namely LLMEmb and LLMExpl. Furthermore, we include LLMPred, a pure LLM instruction-tuning approach that fine-tunes the LLM to directly predict labels without a GNN component.
For the node regression task, following the RelBench benchmark where the graph is heterogeneous, we adopt the corresponding heterogeneous graph versions of GNNs (HeteroGCN, HeteroSAGE, HeteroGAT) and their LLM-enhanced variants (LLMEmb, LLMExpl). Our method (\textit{ours}) is evaluated under the same settings for both tasks. To ensure fair comparison, all experiments are run on the same local machine using self-hosted LLMs to eliminate API-induced variances. Implementation details and hyperparameter settings are provided in Appendix~\ref{sec: implementation details}.

\subsection{Overall Performance Comparison (Q1)}

\input{tables/clf_exp.tex}

\input{tables/reg_exp.tex}

We report the main results across six real-world TAG datasets involving node classification and node regression tasks, with results summarized in Table~\ref{tab:clf-exp} and Table~\ref{tab:reg-exp}.
Our empirical results demonstrate that \method{} runs up to \textbf{167.9$\times$ faster} than existing GraphLLM methods while consistently delivering \textbf{on-par or superior performance}.

\textbf{Comparison with Traditional Baselines:} 
This category includes traditional GNNs (GCN, SAGE, GAT) and pretrained LM baselines (BERT, RoBERTa, SenBERT).
\textit{\underline{Observation 1:} Traditional GNNs are efficient but suffer from limited semantic depth.} While models like GraphSAGE and GAT exhibit the lowest runtimes (averaging under 17s), their accuracy on classification tasks is significantly lower than LLM-augmented methods, with a gap exceeding 5\% on datasets like \textit{Computer}. This indicates that shallow features cannot fully capture the complex attributes in TAGs.
\textit{\underline{Observation 2:} \method{} outperforms pretrained LM methods in both accuracy and efficiency.} Pretrained LM baselines such as RoBERTa and BERT require substantial processing time (e.g., 1832.97s for RoBERTa) without yielding competitive results. In contrast, \method{} is not only faster but also provides a significant boost in predictive power by leveraging LLM-derived explanations rather than raw embeddings. For instance, in regression tasks, \method{} reduces the average MAE from 6.39 (HeteroGCN) to 5.40, demonstrating the necessity of high-level reasoning.

\textbf{Comparison with LLM-based Methods:} 
We compare \method{} against state-of-the-art GraphLLM frameworks, including TAPE, ENGINE, LLaGA, and several LLM-based baselines.
\textit{\underline{Observation 3:} \method{} resolves the computational bottleneck of LLM-augmented graph mining.} Representative methods like TAPE and LLMExpl are severely limited by the cost of repeated LLM queries, requiring over 100,000 seconds for medium-sized benchmarks. \method{} runs \textbf{167.9$\times$ faster} than TAPE while maintaining identical accuracy. This shows that our adaptive sparse querying strategy avoids redundant computations without sacrificing the "LLM gain".
\textit{\underline{Observation 4:} Superior performance in regression tasks through sparse optimization.} As shown in Table~\ref{tab:reg-exp}, \method{} achieves the \textbf{best MAE (Rank 1)} on every regression dataset. Unlike LLMExpl, which queries all nodes, our bilevel-optimized querying selectively identifies nodes where LLM insights are most beneficial. This sparse strategy not only saves time but also filters out potentially noisy or low-gain textual information, leading to more precise numerical predictions.
\textit{\underline{Observation 5:} Practicality and Scalability.} While efficiency-focused models like ENGINE and LLaGA reduce runtime to some extent, \method{} remains the fastest among all LLM-based competitors (ranking 1st in speed in both Table~\ref{tab:clf-exp} and Table~\ref{tab:reg-exp}). By bringing the total runtime closer to standard GNN training, \method{} stands out as the most viable framework for deploying GraphLLM in real-world, large-scale scenarios.

\subsection{Scaling up to Million-Scale Graph (Q2)}
\label{subsec:large-scale}

To evaluate the practical scalability of \method{}, we conduct experiments on the \textit{ogbn-products} dataset from the Open Graph Benchmark (OGB)~\cite{hu2020b}. This is a large-scale text-attributed graph comprising over 2.4 million nodes and 123.7 million edges, with an average of 145.50 tokens per node, posing a significant computational challenge for LLM-based methods. To ensure a representative comparison, we select the most efficient baselines from each category: GCN and SenBERT as classic GNN and LM-based methods, and LLaGA, LLMExpl, and LLMEmb as representative GraphLLM frameworks.

\input{tables/large_exp.tex}

The results, summarized in Table~\ref{tab:large-exp}, reveal a stark contrast in efficiency. While classic methods (GCN, SenBERT) are computationally efficient, their performance is limited by the lack of high-level reasoning. Conversely, existing GraphLLM methods, despite their potential, become computationally intractable at this scale, failing to complete within a 24-hour time limit due to the exhaustive nature of their LLM querying or embedding processes. In contrast, \method{} successfully processes the entire graph in just 7.4 hours, achieving the highest accuracy of 78.42\%. This demonstrates that our bilevel-optimized sparse querying strategy effectively breaks the scalability bottleneck, making LLM-enhanced graph learning viable for industrial-scale applications. Additional large-scale results on arXiv are reported in Appendix~\ref{sec: arxiv-scalability}.

%% file: tables/clf_exp.tex
\begin{table*}[!t]
\centering
\caption{\textbf{Performance comparison on node classification tasks with Accuracy (\%) and Total Time (s) reported.} The \colorbox{orange!25}{\textbf{best}} and \colorbox{orange!10}{second-best} accuracy results, as well as the \colorbox{blue!25}{\textbf{fastest}} and \colorbox{blue!10}{second-fastest} times, are highlighted.}
\label{tab:clf-exp}
\small
\resizebox{\textwidth}{!}{%
\begin{tabular}{lcccccccc}
\hline
\multicolumn{1}{c}{\multirow{2}{*}{Methods}} & \multicolumn{2}{c}{Photo} & \multicolumn{2}{c}{Instagram} & \multicolumn{2}{c}{Computer} & \multicolumn{2}{c}{Average} \\ \cline{2-9} 
\multicolumn{1}{c}{} & time (s) \(\downarrow\)   & acc  \(\uparrow\)      & time (s) \(\downarrow\)  & acc  \(\uparrow\)      & time (s) \(\downarrow\)   & acc        & time (s) \(\downarrow\)  & acc  \(\uparrow\) \\ \hline
GCN                                          & 15.15               & 70.63{\tiny ±0.71}          & 1.78            & 63.32{\tiny ±0.22}          & 21.47               & 71.70{\tiny ±0.77}          & 12.80              & 68.55          \\
SAGE                                         & 15.62               & 76.44{\tiny ±0.23}          & 1.72            & 63.22{\tiny ±0.41}          & 21.67               & 80.33{\tiny ±0.47}          & 13.00              & 73.33          \\
GAT                                          & 18.05               & 79.63{\tiny ±0.39}          & 2.27            & 63.39{\tiny ±0.22}          & 27.86               & 81.93{\tiny ±1.40}          & 16.06              & 74.98          \\ 
SenBERT                                  & 321.44              & 74.14{\tiny ±0.18}          & 40.08           & 62.71{\tiny ±0.58}          & 578.36              & 70.34{\tiny ±0.25}          & 313.29             & 69.06          \\
RoBERTa                                & 1,787.73            & 75.15{\tiny ±0.34}          & 427.19          & 65.33{\tiny ±0.37}          & 3,283.98            & 71.02{\tiny ±0.61}          & 1832.97            & 70.50          \\
BERT                                         & 637.51              & 73.11{\tiny ±0.09}          & 71.98           & 62.55{\tiny ±0.39}          & 1,162.61            & 68.35{\tiny ±0.62}          & 624.03             & 68.00          \\ \hline \hline
TAPE                                         & 120,802.00 & \cellcolor{orange!25} \textbf{85.79{\tiny ±0.13}} & 31,475.80       & 65.74{\tiny ±1.25}          & 149,887.00 & \cellcolor{orange!25} \textbf{89.73{\tiny ±0.06}} & 100721.60 & \cellcolor{orange!25} \textbf{80.42} \\
ENGINE                                       &  2,613.60            & 84.78{\tiny ±0.20}          &  536.01    & \cellcolor{orange!10} 67.09{\tiny ±0.38}    & 4,196.57            & 87.06{\tiny ±0.23}          & 2448.73            & 79.64          \\
LLaGA                                        & 4,288.58            & 84.66{\tiny ±0.17}          & 581.77          & 64.85{\tiny ±1.17}          & 7,445.99            & 88.02{\tiny ±0.49}          & 4105.45            & 79.18          \\
LLMExpl                                      & 116,157.00          & 84.97{\tiny ±0.25}          & 30,737.90       & 64.94{\tiny ±0.24}          & 144,291.00          & 88.39{\tiny ±0.10}          & 97061.97           & 79.43          \\
LLMEmb                                       & \cellcolor{blue!10} 2,077.30            & 85.35{\tiny ±0.24}          & \cellcolor{blue!10}  334.94 & \cellcolor{orange!25} \textbf{67.21{\tiny ±0.34}} & \cellcolor{blue!10} 3,293.02            & 88.60{\tiny ±0.15}          & \cellcolor{blue!10} 1901.75            & 80.39          \\
LLMPred                                      & 13,515.00           & 73.69{\tiny ±0.37}          & 11,794.40       & 43.69{\tiny ±0.69}          & 29,592.80           & 70.39{\tiny ±0.12}          & 18300.73           & 62.59          \\
ours                                         & \cellcolor{blue!25} \textbf{490.55}        & \cellcolor{orange!10} { 85.75{\tiny ±0.28}}    & \cellcolor{blue!25} \textbf{129.04}          &  66.22{\tiny ±0.55}          & \cellcolor{blue!25} \textbf{1,180.41}      & \cellcolor{orange!10} { 89.30{\tiny ±0.11}}    & \cellcolor{blue!25} \textbf{600.00}       & \cellcolor{orange!25} \textbf{80.42}    \\ \hline
\end{tabular}
}
\end{table*}

%% file: tables/reg_exp.tex
\begin{table*}[!t]
\centering
\caption{\textbf{Performance comparison on node regression tasks with Mean Absolute Error (MAE) and Total Time (s) reported.} The \colorbox{orange!25}{\textbf{best}} and \colorbox{orange!10}{second-best} MAE results, as well as the \colorbox{blue!25}{\textbf{fastest}} and \colorbox{blue!10}{second-fastest} times, are highlighted.}
\label{tab:reg-exp}
\small
\resizebox{\textwidth}{!}{%
\begin{tabular}{lcccccccc}
\hline
\multirow{2}{*}{Methods} & \multicolumn{2}{c}{User-ltv}                 & \multicolumn{2}{c}{Item-ltv}                  & \multicolumn{2}{c}{Post-votes}               & \multicolumn{2}{c}{Average}        \\ \cline{2-9} 
                         & time (s) \(\downarrow\)        & mae \(\downarrow\)           & time (s) \(\downarrow\)        & mae \(\downarrow\)            & time (s) \(\downarrow\)        & mae \(\downarrow\)           & time (s) \(\downarrow\)        & mae \(\downarrow\) \\ \hline
HeteroGCN                & 13.71           & 1.39{\tiny ±0.06}          & 13.67           & 17.67{\tiny ±0.29}          & 15.77           & 0.11{\tiny ±0.03}          & 14.38           & 6.39             \\
HeteroSAGE               & 13.82           & 1.37{\tiny ±0.01}          & 13.77           & 17.77{\tiny ±0.30}          & 16.18           & 0.12{\tiny ±0.04}          & 14.59           & 6.42             \\
HeteroGAT                & 13.74           & 1.39{\tiny ±0.06}          & 13.68           & 17.67{\tiny ±0.29}          & 15.76           & 0.11{\tiny ±0.03}          & 14.39           & 6.39             \\ \hline \hline
LLMEmb                   & \cellcolor{blue!10} 1,430.20        & 1.35{\tiny ±0.03}          & \cellcolor{blue!10} 1,434.41        & 15.64{\tiny ±0.19}          & \cellcolor{blue!10} 1,713.90        & 0.10{\tiny ±0.03}          & \cellcolor{blue!10} 1526.17         & 5.70             \\
LLMExpl                  & 32,672.90       & \cellcolor{orange!10} 1.34{\tiny ±0.03}          & 34,600.20       & \cellcolor{orange!10} 15.40{\tiny ±0.07}          & 32,230.40       & \cellcolor{orange!10} 0.09{\tiny ±0.03}          & 33167.83        & \cellcolor{orange!10} 5.61             \\
ours                     & \cellcolor{blue!25} \textbf{419.03} & \cellcolor{orange!25} \textbf{1.33{\tiny ±0.02}} &\cellcolor{blue!25} \textbf{421.62} & \cellcolor{orange!25} \textbf{14.79{\tiny ±0.01}} & \cellcolor{blue!25} \textbf{494.32} & \cellcolor{orange!25} \textbf{0.08{\tiny ±0.02}} & \cellcolor{blue!25} \textbf{444.99} & \cellcolor{orange!25} \textbf{5.40}    \\ \hline
\end{tabular}
}
\end{table*}

%% file: tables/large_exp.tex
\begin{table}[h] 
\centering
\vspace{-2mm}
\caption{\textbf{Scalability and efficiency analysis on the million-scale ogbn-products dataset.} '---' indicates the method exceeded the 24-hour time limit. \method{} achieves the best accuracy while remaining computationally feasible.}
\label{tab:large-exp}
\small
\begin{tabular}{lccc}
\toprule
\textbf{Methods} & \textbf{Accuracy (\%)} & \textbf{Time (hrs)} & \textbf{Efficiency} \\ 
\midrule
GCN     & 69.01 & 0.1 & \checkmark \\
SenBERT & 77.10 & 2.9 & \checkmark \\ 
\midrule
LLaGA   & ---   & $>24$ & \ding{55} \\
LLMExpl & ---   & $>24$ & \ding{55} \\
LLMEmb  & ---   & $>24$ & \ding{55} \\
\rowcolor[gray]{0.9} \textbf{Ours} & \textbf{78.42} & \textbf{7.4} & \checkmark \\ 
\bottomrule
\end{tabular}
\end{table}

%% file: sections/02-rw.tex
\section{Related Work}

\textbf{Graph Neural Networks.} GNNs~\cite{kipf2017,velickovic2018} are standard for graph learning, and TAG methods increasingly use PLM embeddings~\cite{devlin2019,liu2019a} or PLM-GNN interaction~\cite{zhao2023,jin2023a}. However, prior studies note that these designs can still incur prohibitive computation~\cite{he2024,zhu2024b}.

\textbf{Graph with Large Language Models.} Existing GraphLLM methods for node-level tasks fall into three categories based on LLM roles~\cite{wu2025a} : (1) LLM as predictor, linearizing entire graphs for direct LLM predictions (e.g., LLaGA~\cite{chen2024c}, GraphGPT~\cite{tang2024b}); (2) LLM as encoder, jointly embedding node text and context (e.g., ENGINE~\cite{zhu2024b}); (3) LLM as explainer (our focus), where LLMs provide auxiliary explanation features to enhance GNN inputs (e.g., TAPE~\cite{he2024}). 

\textbf{Bilevel Optimization in Machine Learning.} Bilevel optimization addresses nested problems where an outer objective depends on an inner solution~\cite{zhang2024j}, and is widely used in hyperparameter tuning~\cite{lorraine2020} and meta-learning~\cite{finn2017}. We pioneer the use of bilevel optimization in the \emph{efficient GraphLLM} domain, focusing on selective LLM querying in graphs. This formulation directly aligns query selection with task utility and inherently promotes sparsity, addressing both efficiency and effectiveness.

\textbf{Node-level Tasks on TAGs.} Most existing works on TAGs focus primarily on \textit{node classification}~\cite{chen2022a,chen2025a,deng2024}. However, with the growing interest in representation learning for Relational Deep Learning tasks, \textit{node regression}—predicting continuous node attributes—has attracted increasing attention due to its practical relevance in domains such as e-commerce systems~\cite{robinson2024}. To validate the generality of \method{}, we conduct experiments on both task types across diverse real-world TAG datasets.

%% file: sections/05-conclusion.tex
\vspace{-2mm}
\section{Conclusion}

We proposed \method{}, a bilevel-optimized sparse querying framework that reduces the scalability bottleneck of GraphLLMs by invoking LLMs only for task-critical nodes. Across classification, regression, and million-scale graph experiments, \method{} runs substantially faster while maintaining competitive or superior predictive performance, making LLM-enhanced graph learning more practical.

%% file: sections/06-appendix.tex
\section{Derivation of Sparse Query Hypergradient}\label{prf:ift}

\begin{proof}[Derivation]
We derive the sparse-query hypergradient in Equation~\ref{eq:sparse-hypergradient} here. For a given selector score \(\boldsymbol{\lambda}'\), let \(\mathbf{w}' = \mathbf{w}^*(\boldsymbol{\lambda}')\) denote the corresponding inner-loop solution. We assume that \(\mathbf{w}'\) satisfies the stationarity condition
\[
\left.\frac{\partial \mathcal{L}_T}{\partial \mathbf{w}}\right|_{(\boldsymbol{\lambda}', \mathbf{w}')} = 0,
\]
and that the training Hessian \(\mathbf{H}_{\mathbf{w}\mathbf{w}} := \frac{\partial^2 \mathcal{L}_T}{\partial \mathbf{w} \partial \mathbf{w}^T}\) is locally invertible at \((\boldsymbol{\lambda}', \mathbf{w}')\).

Define the vector-valued function
\[
F(\boldsymbol{\lambda}, \mathbf{w}) := \frac{\partial \mathcal{L}_{T}}{\partial \mathbf{w}}(\boldsymbol{\lambda}, \mathbf{w}).
\]
By assumption, \(F(\boldsymbol{\lambda}', \mathbf{w}') = 0\). The Implicit Function Theorem states that if the Jacobian of \(F\) with respect to \(\mathbf{w}\) at \((\boldsymbol{\lambda}', \mathbf{w}')\), i.e.,
\[
\begin{aligned}
J_{\mathbf{w}} F(\boldsymbol{\lambda}', \mathbf{w}')
&:= \frac{\partial F}{\partial \mathbf{w}}(\boldsymbol{\lambda}', \mathbf{w}') \\
&= \frac{\partial^{2} \mathcal{L}_{T}}
{\partial \mathbf{w} \partial \mathbf{w}^T}
(\boldsymbol{\lambda}', \mathbf{w}'),
\end{aligned}
\]
is invertible, then there exists a neighborhood of \(\boldsymbol{\lambda}'\) and a unique differentiable function \(\mathbf{w}^*(\boldsymbol{\lambda})\) such that
\[
F(\boldsymbol{\lambda}, \mathbf{w}^*(\boldsymbol{\lambda})) = 0.
\]

Differentiating this identity with respect to \(\boldsymbol{\lambda}\) yields
\[
\frac{\partial F}{\partial \boldsymbol{\lambda}} + \frac{\partial F}{\partial \mathbf{w}} \frac{\partial \mathbf{w}^*}{\partial \boldsymbol{\lambda}} = 0.
\]
Solving for the Jacobian of the optimal parameters gives
\[
\frac{\partial \mathbf{w}^*}{\partial \boldsymbol{\lambda}} = - \left( \frac{\partial F}{\partial \mathbf{w}} \right)^{-1} \times \frac{\partial F}{\partial \boldsymbol{\lambda}}.
\]

Substituting back \(F = \partial \mathcal{L}_T / \partial \mathbf{w}\), we have:
\[
\frac{\partial \mathbf{w}^*}{\partial \boldsymbol{\lambda}} = - \left[ \frac{\partial^{2} \mathcal{L}_{T}}{\partial \mathbf{w} \partial \mathbf{w}^T} \right]^{-1} \times \frac{\partial}{\partial \boldsymbol{\lambda}} \left( \frac{\partial \mathcal{L}_{T}}{\partial \mathbf{w}} \right).
\]

To align with the notation in Equation~\ref{eq:sparse-hypergradient}, we note that the matrix \(\frac{\partial}{\partial \boldsymbol{\lambda}} \left( \frac{\partial \mathcal{L}_{T}}{\partial \mathbf{w}} \right)\) has entries \(\left[ \frac{\partial}{\partial \boldsymbol{\lambda}} \left( \frac{\partial \mathcal{L}_{T}}{\partial \mathbf{w}} \right) \right]_{i,j} = \frac{\partial^2 \mathcal{L}_T}{\partial w_i \partial \lambda_j}\). This matrix is commonly denoted by \(\frac{\partial^2 \mathcal{L}_{T}}{\partial \mathbf{w} \partial \boldsymbol{\lambda}^T}\). Therefore,
\[
\frac{\partial \mathbf{w}^*}{\partial \boldsymbol{\lambda}} 
= - \left[ \frac{\partial^{2} \mathcal{L}_{T}}{\partial \mathbf{w} \partial \mathbf{w}^T} \right]^{-1} \times \frac{\partial^{2} \mathcal{L}_{T}}{\partial \mathbf{w} \partial \boldsymbol{\lambda}^T}.
\]

Since the exact inverse of the Hessian is computationally expensive, it can be approximated by an infinite series:

\begin{equation}
\left[ \frac{\partial^2 \mathcal{L}_{T}}{\partial \mathbf{w} \partial \mathbf{w}^T} \right]^{-1}
= 
\lim_{i \to \infty} \sum_{j=0}^i \left( I - \frac{\partial^2 \mathcal{L}_{T}}{\partial \mathbf{w} \partial \mathbf{w}^T} \right)^j 
\end{equation}

Substitute it into the above equation to get the final result:

\[
\frac{\partial \mathbf{w}^*}{\partial \boldsymbol{\lambda}} 
= - \lim_{i \to \infty} \sum_{j=0}^i \left( I - \frac{\partial^2 \mathcal{L}_{T}}{\partial \mathbf{w} \partial \mathbf{w}^T} \right)^j \times \frac{\partial^{2} \mathcal{L}_{T}}{\partial \mathbf{w} \partial \boldsymbol{\lambda}^T}.
\]

Since \(\mathcal{L}_{V}^{*}(\boldsymbol{\lambda}) = \mathcal{L}_{V}(\mathbf{w}^{*}(\boldsymbol{\lambda}))\) under our validation objective, the direct derivative with respect to \(\boldsymbol{\lambda}\) is zero and the chain rule gives
\[
\frac{\partial \mathcal{L}_{V}^{*}}
{\partial \boldsymbol{\lambda}}
=
\frac{\partial \mathcal{L}_{V}}
{\partial \mathbf{w}^{*}}
\frac{\partial \mathbf{w}^{*}}
{\partial \boldsymbol{\lambda}}.
\]
Substituting the Neumann-series approximation above yields the stated sparse-query hypergradient.
This completes the derivation.
\end{proof}

\section{Dataset Details}\label{sec: appendix-dataset}

\subsection{Node Classification Datasets}

The Ele-Photo (abbreviated as \textbf{Photo}) and Ele-Computer (abbreviated as \textbf{Computer}) datasets are derived from the Amazon Electronics dataset~\cite{ni2019}, where each node represents an item in either the Photo or Computer category. In these e-commerce networks, edges indicate co-purchase or co-view relationships between items. The associated text for each item comprises product descriptions, such as summaries or user reviews. The classification task involves categorizing these products into fine-grained sub-categories.
The \textbf{Instagram} dataset, originally released in~\cite{huang2024e}, features nodes representing users, with edges denoting social connections such as following relationships. Each Instagram node includes textual features extracted from the user’s profile page introduction. The classification task labels specify whether a user is a commercial or a normal user.

Table~\ref{tab:label-space} summarizes the label space for each dataset.

\begin{table*}[t]
\centering
\caption{Label space of classification datasets.}
\label{tab:label-space}
\small
\setlength{\tabcolsep}{4pt}
\begin{tabular}{p{0.12\textwidth}|p{0.82\textwidth}}
\hline
\textbf{Dataset} & \multicolumn{1}{c}{\textbf{Label Space}}                          \\ \hline
\multirow{3}{*}{Photo}    & Lighting \& Studio, Bags \& Cases, Tripods \& Monopods, Flashes,              \\
                 & Video Surveillance, Accessories, Binoculars \& Scopes, Video,     \\
                 & Digital Cameras, Film Photography, Lenses, Underwater Photography \\ \hline
\multirow{3}{*}{Computer} & Computer Accessories \& Peripherals, Tablet Accessories, Laptop Accessories,  \\
                          & Computers \& Tablets, Computer Components, Data Storage, Networking Products, \\
                 & Monitors, Servers, Tablet Replacement Parts                       \\ \hline
Instagram        & Normal User, Commercial User                                      \\ \hline
\end{tabular}
\end{table*}

\textbf{Train-Test Split.} We follow a semi-supervised learning setup, where only a small subset of nodes is labeled to imitate real-world scenarios with limited annotation. This setting evaluates the model’s ability to leverage scarce labeled data effectively. Specifically, we adopt a 10\% training data split on Instagram, Photo, and Computer datasets, consistent with the experimental protocol used in LLMNodeBed~\cite{wu2025a}.

\subsection{Node Regression Datasets}

\textbf{User-ltv} is derived from the Rel-amazon dataset, which contains comprehensive records of products, users, and reviews on Amazon’s platform focusing on book-related products. Each node represents a user, enriched with textual information from their reviews and purchase history. In this regression task, the objective is to predict the total monetary value (\$) of all products a user will buy and review over the next 3 months.
\textbf{Item-ltv} is also extracted from the same Rel-amazon subset, where nodes correspond to individual products in the book category. Product nodes include textual features such as product descriptions and user reviews. The regression target is to forecast the total monetary value (\$) of purchases and reviews that each product will receive in the coming 3 months.
\textbf{Post-votes} originates from the Stack Exchange stats-exchange site. Here, nodes represent individual user posts, accompanied by rich textual content such as raw text in posts and comments, as well as metadata like edit histories and voting records. The regression task \textit{post-votes} aims to predict the number of votes each post will receive in the next 3 months based on its current state and textual information.

\textbf{Data Downsampling.} Random downsampling is applied to all three node regression datasets to reduce size while preserving representativity. Isolated nodes caused by this process are subsequently removed to maintain graph connectivity.

\textbf{Train-Test Split.} We follow the RelBench~\cite{robinson2024} temporal splitting scheme: training data includes all records up to the validation timestamp (VAL\_TIMESTAMP), validation data covers from VAL\_TIMESTAMP to the test timestamp (TEST\_TIMESTAMP), and testing data includes all records after TEST\_TIMESTAMP.

\paragraph{Artifact use and data considerations.}
All datasets used in this paper are existing research benchmarks released by their original creators. We use them only for the intended purpose of graph learning evaluation, follow the splits and preprocessing protocols specified by the corresponding benchmark papers, and direct users of our released code to obtain the datasets from the original sources and comply with their licenses, terms of use, and access conditions. We do not create a new dataset, collect new human-subject data, or redistribute raw user-generated text. Because several benchmarks contain product reviews, profile descriptions, posts, comments, or related metadata, they may include personally identifying or offensive content inherited from the original releases. Our experiments treat these fields as benchmark text attributes for model evaluation only, do not attempt to identify individuals, and do not use the derived artifacts outside research contexts.

The detailed statistics and temporal information of these datasets are summarized in Table~\ref{tab:reg-dataset-details}.

\begin{table}[htbp]
\centering
\caption{Detailed Regression Dataset information}
\label{tab:reg-dataset-details}
\small
\resizebox{\columnwidth}{!}{%
\begin{tabular}{lccc}
\toprule
 & \textbf{User-ltv} & \textbf{Item-ltv} & \textbf{Post-votes} \\ 
\midrule
\# Nodes      & 35,772  & 35,772  & 40,947  \\
\# Edges      & 70,632  & 70,632  & 91,142  \\
\# node types & 3       & 3       & 7       \\
\# edge types & 4       & 4       & 22      \\
Start Timestamp & 2008-01-01 & 2008-01-01 & 2009-02-02 \\
Val Timestamp & 2015-10-01 & 2015-10-01 & 2020-10-01 \\
Test Timestamp & 2016-01-01 & 2016-01-01 & 2021-01-01 \\
End Timestamp & 2018-09-28 & 2018-09-28 & 2023-09-03 \\
\bottomrule
\end{tabular}
}
\end{table}

\section{Implementation Details}\label{sec: implementation details}
In this section, we provide the detailed experimental setup, including the hyperparameter search space, model-specific configurations, and the prompting strategy used for LLM querying.

\subsection{Hyperparamter settings}

\textbf{Common Experimental Setup.} To ensure a rigorous and fair evaluation of computational efficiency, all experiments are conducted on a Linux server equipped with a NVIDIA 48GB L40 GPU and an Intel(R) Xeon(R) Platinum 8358 CPU (64 cores). Our code is developed based on LLMNodeBed~\cite{wu2025a}. Notably, rather than relying on commercial cloud-based APIs (e.g., GPT-4), which are subject to unpredictable network latency and heterogeneous underlying hardware, we employ locally deployed LLMs for all text generation tasks. This setup ensures that the reported faster runtimes are attributable to the algorithmic improvements of our adaptive sparse querying strategy rather than external infrastructure variances. We use Mistral-7B as the default LLM as done in LLMNodeBed; this choice provides a controlled local efficiency setting rather than a restriction of the method itself. Additional backbone robustness results are reported in Appendix~\ref{sec:llm-backbone-robustness}.

\textbf{Node Classification Task.} We perform a grid search to optimize the performance of each model. We denote the number of layers, hidden dimension, and learning rate by $L$, $H$, and $\eta$, respectively.

\begin{itemize}

\item \textbf{GNNs:} The number of GNN layers is searched among $\{1, 2, 3\}$, and the hidden dimension is selected from $\{128, 256\}$. The learning rate is tuned within $\{0.005, 0.01\}$. Following LLMNodeBed, we set the maximum number of training epochs to 500 with an early stopping patience of 100. We find the best performance with $(L,H,\eta)=(3,128,0.005)$ for GCN on Photo and Computer, $(1,128,0.005)$ for GCN on Instagram, $(3,128,0.01)$ for GraphSAGE on Photo, $(3,128,0.005)$ for GraphSAGE on Computer and Instagram, $(3,128,0.005)$ for GAT on Photo, $(3,128,0.01)$ for GAT on Computer, and $(2,256,0.01)$ for GAT on Instagram.

\item \textbf{Pretrained LMs:} Due to the high computational cost of performing a grid search for fine-tuning language models, we adopt the default hyperparameters from LLMNodeBed, specifically using 10 epochs and a learning rate of $2 \times 10^{-5}$.

\item \textbf{ENGINE:} Following the original paper, we search for the number of layers in $\{1, 2, 3\}$, the hidden dimension in $\{64, 128\}$, and the learning rate in $\{5 \times 10^{-4}, 1 \times 10^{-3}\}$. We find the best performance with $(L,H,\eta)=(2,128,5\times10^{-4})$ on Photo, Computer, and Instagram.

\item \textbf{TAPE:} We utilize Mistral-7B for explanation generation. To ensure reproducibility and manage computational overhead, we use the pre-generated explanations provided by LLMNodeBed. The total explanation generation time is estimated by multiplying the total token count by the per-token latency measured on our local machine. For the GNN component, we search the number of layers in $\{2, 3, 4\}$, hidden dimensions in $\{128, 256\}$, and learning rates in $\{5 \times 10^{-4}, 1 \times 10^{-3}\}$. The LM (RoBERTa-large, 355M) is fine-tuned using its default parameters. We find the best performance with $(L,H,\eta)=(2,128,10^{-3})$ on Photo, $(4,256,5\times10^{-4})$ on Computer, and $(4,256,10^{-3})$ on Instagram.

\item \textbf{LLaGA:} Following the settings of LLMNodeBed, we adopt the HO templates as the default configuration with the number of hops set to 4. RoBERTa-355M is employed as the text encoder, and the linear projection layer is implemented as a 2-layer MLP with a hidden dimension of 2048. We set the batch size to 64 and the learning rate to $10^{-4}$. The number of training epochs is 10 for Instagram and 8 for the larger Computer and Photo datasets.

\item \textbf{LLMEmb:} This method generates node embeddings by mean-pooling the hidden states from the final layer of the LLM. For the GNN backbone, we perform a grid search over the number of layers in $\{1, 2, 3\}$, hidden dimensions in $\{128, 256\}$, and learning rates in $\{0.005, 0.01\}$. We find the best performance with $(L,H,\eta)=(3,256,0.005)$ on Photo and Computer, and $(2,128,0.005)$ on Instagram.

\item \textbf{LLMExpl:} This approach concatenates the LLM-generated explanations with the original text attributes and subsequently utilizes a frozen LM (RoBERTa-large, 355M) to encode the augmented text into node embeddings. The GNN backbone's hyperparameters are tuned across the same search space as LLMEmb: layers in $\{1, 2, 3\}$, hidden dimensions in $\{128, 256\}$, and learning rates in $\{0.005, 0.01\}$. We find the best performance with $(L,H,\eta)=(3,256,0.005)$ on Photo and Computer, and $(1,256,0.005)$ on Instagram.

\item \textbf{LLMPred:} This baseline corresponds to the LLM Instruction Tuning approach from LLMNodeBed. The method directly fine-tunes the LLM to predict node labels without incorporating any GNN components. We follow the default hyperparameters specified in LLMNodeBed: the LoRA rank $r$ is set to 8, and the scaling factor $\alpha$ is set to 16. We use a dropout ratio of 0.1 and a learning rate of $10^{-5}$. For each dataset, the input consists of the node's original text concatenated with a task-specific prompt designed to guide the LLM in classification, while the expected output is the categorical label. We train the model for 8 epochs on Instagram and 2 epochs on Computer and Photo, adjusting for the larger scale of the latter two datasets.

\item \textbf{\method{} (Ours):}  For our framework, the GNN backbone's hyperparameter search space includes the number of layers in $\{2, 3\}$ and hidden dimensions in $\{128, 256\}$, with a fixed learning rate of 0.01. For the bilevel optimization, we search the temperature $\tau$ within $\{2.0, 4.0\}$ and the multiplier $\lambda$'s learning rate within $\{10^{-3}, 10^{-2}\}$. We fix the sparsity $k=10$ and the minimum temperature $\tau_{min}=0.5$. We conduct a sensitivity analysis on the hyperparameter k in Appendix~\ref{app:sensitivity_k}. The results show that 10 is a reasonable value. Consistent with standard practice, we employ RoBERTa-large (355M) as the default language model to encode text to node embeddings. The outer loop number of iterations $T$ is set to 3, and the inner loop number of iterations $I$ is set to 200 and early stopping patience of 50. The Neumann series approximation uses $i=10$ steps. We find the best performance with $(L,H,\eta,\tau,\eta_{\lambda})=(3,256,0.01,4.0,10^{-2})$ on Photo and Computer, and $(3,128,0.01,4.0,10^{-3})$ on Instagram.

\end{itemize}

\textbf{Node Regression Task.} For regression, the GNN baselines, LLMEmb, and LLMExpl follow the same hyperparameter search space as the node classification task. For \method{}, we set the sparsity $k=10$, temperature $\tau=4.0$, $\tau_{\min}=0.5$, and the learning rate for $\lambda$ to $0.005$. The GNN backbone hyperparameters are searched over $\{1, 2, 3\}$ layers, $\{128, 256\}$ hidden dimensions, and a learning rate range of $\{0.005, 0.01\}$. For GNNs, we find the best performance with $(L,H,\eta)=(2,128,0.005)$ on user-ltv, item-ltv, and post-votes. For LLMEmb and LLMExpl, we find the best performance with $(L,H,\eta)=(2,128,0.005)$ on user-ltv, $(2,128,0.01)$ on item-ltv, and $(3,128,0.005)$ on post-votes. For \method{}, we find the best performance with $(L,H,\eta)=(2,128,0.005)$ on user-ltv, $(2,256,0.01)$ on item-ltv, and $(1,128,0.01)$ on post-votes.

\textbf{Analysis of Text Augmentation Strategy.} In the preliminary stage, we compared two ways to incorporate LLM-derived explanations: (1) concatenating the original text $s_n$ with the explanation $e_n$, and (2) replacing $s_n$ with $e_n$. Our empirical results indicated that direct replacement ($s_n' \leftarrow e_n$) yields superior performance. We speculate that concatenation may introduce noise or redundant information from the raw text, which could dilute the task-specific signals present in the LLM-generated explanations during the embedding process.

\subsection{Additional Large-Scale Results on arXiv}
\label{sec: arxiv-scalability}

To further evaluate scalability beyond ogbn-products, we conduct an additional experiment on the arXiv citation graph, which contains 169,343 nodes and 1,166,243 edges, with an average of 216.95 tokens per raw text. As shown in Table~\ref{tab:arxiv-scalability}, \method{} remains close to the strongest LLM-based baselines in accuracy while requiring substantially less runtime. This result supports the main conclusion that sparse querying improves the efficiency-performance trade-off on large text-attributed graphs.

\begin{table}[htbp]
\centering
\caption{Additional scalability results on arXiv. Runtime is reported in seconds.}
\label{tab:arxiv-scalability}
\small
\begin{tabular}{lcc}
\toprule
Method & Time (s) $\downarrow$ & Accuracy (\%) $\uparrow$ \\
\midrule
GCN & 26.71 & 70.30 \\
SenBERT & 3355.02 & 72.61 \\
LLaGA & 37208.20 & 73.99 \\
LLMExpl & 522727.04 & 73.37 \\
LLMEmb & 9125.84 & 73.97 \\
\rowcolor[gray]{0.9} \textbf{Ours} & \textbf{1363.36} & 73.89 \\
\bottomrule
\end{tabular}
\end{table}

\subsection{Robustness to LLM Backbones}
\label{sec:llm-backbone-robustness}

Table~\ref{tab:llm-backbone} evaluates \method{} with different LLM backbones for explanation generation. Accuracy remains stable across local and API-served models, suggesting that the selection strategy does not depend on a specific LLM architecture. For API-served models, runtime is omitted because it is confounded by network and service latency rather than local algorithmic cost.

\begin{table*}[htbp]
\centering
\caption{Robustness to different LLM backbones. Runtime is reported only for locally deployed models.}
\label{tab:llm-backbone}
\small
\resizebox{\textwidth}{!}{%
\begin{tabular}{lcccccc}
\toprule
Backbone & Photo Time (s) & Photo Acc. & Instagram Time (s) & Instagram Acc. & Computer Time (s) & Computer Acc. \\
\midrule
Mistral-7B (local) & 490.55 & 85.75 $\pm$ 0.28 & 129.04 & 66.22 $\pm$ 0.55 & 1180.41 & 89.30 $\pm$ 0.11 \\
Qwen3-8B (local) & 449.34 & 85.76 $\pm$ 0.27 & 150.66 & 65.89 $\pm$ 0.50 & 682.20 & 89.28 $\pm$ 0.13 \\
Qwen3-27B (openrouter api) & -- & 85.77 $\pm$ 0.20 & -- & 65.88 $\pm$ 0.30 & -- & 89.25 $\pm$ 0.12 \\
MiniMax-M2.7 (openrouter api) & -- & 85.80 $\pm$ 0.19 & -- & 65.81 $\pm$ 0.49 & -- & 89.27 $\pm$ 0.12 \\
\bottomrule
\end{tabular}
}
\end{table*}

\subsection{LLM Querying and Output Length Control}
In line with prior observations that large language models often generate overly verbose outputs~\cite{nayab2025}, we control the response length through meticulous prompt design to improve conciseness and computational efficiency.
Specifically, we instruct the LLM to limit its reasoning to at most three potential categories for classification tasks or three representative numerical values for regression tasks, and to restrict the total answer length within a fixed budget of $\mathcal{B}$ words. In our implementation, $\mathcal{B} = 50$.

\subsection{Prompts}

To ensure the reproducibility and consistency of the explanation features, we employ a unified prompt structure across all datasets. The prompt consists of a task-specific question followed by a fixed set of reasoning and formatting constraints. All prompts follow the template defined below. The placeholder [Task-specific
Question] is replaced by the corresponding query for each dataset as listed in Table~\ref{tab:prompt_questions}.

\begin{promptbox}{General Prompt Structure}
\textbf{Question:} [Task-specific Question] If multiple [options/values] apply, provide a comma-separated list ordered from most to least related, then for each choice, explain how it is present in the text. Limit the output to three categories or numerical values and keep the answer within 50 words. \textbf{Answer:}
\end{promptbox}

\begin{table*}[t]
\centering
\caption{Task-specific question components for the unified LLM prompt template.}
\label{tab:prompt_questions}
\small
\setlength{\tabcolsep}{4pt}
\renewcommand{\arraystretch}{1.2}
\begin{tabular}{@{}p{0.20\textwidth}p{0.11\textwidth}p{0.62\textwidth}@{}}
\toprule
Dataset & Output & Task-specific question \\
\midrule
\textit{Amazon User LTV} & Numeric & What is the total \$ value of purchases this customer node will make in the next 3 months? \\
\midrule
\textit{Amazon Item LTV} & Numeric & What is the total \$ value of purchases this product node will receive in the next 3 months? \\
\midrule
\textit{Stack Post Votes} & Numeric & What is the total number of votes this post node will receive in the next 3 months? \\
\midrule
\textit{Instagram} & Category & Which of the following categories does this user on Instagram belong to: Normal Users, Commercial Users? \\
\midrule
\textit{Computer} & Category & Which of the following sub-categories of computer items does this item belong to: Computer Accessories \& Peripherals, Tablet Accessories, Laptop Accessories, Computers \& Tablets, Computer Components, Data Storage, Networking Products, Monitors, Servers, Tablet Replacement Parts? \\
\midrule
\textit{Photo} & Category & Which of the following sub-categories of photo items does this item belong to: Video Surveillance, Accessories, Binoculars \& Scopes, Video, Lighting \& Studio, Bags \& Cases, Tripods \& Monopods, Flashes, Digital Cameras, Film Photography, Lenses, Underwater Photography? \\
\bottomrule
\end{tabular}
\end{table*}

\subsection{Pseudocode}
Algorithm~\ref{alg:hyper-grad} shows the approximate hypergradient computation via Neumann series for efficient bilevel optimization. 
Algorithm~\ref{alg:gumbel-softmax} shows the Gumbel-Softmax Top-$K$ sampling process with the Straight-Through Estimator (STE) to ensure differentiability in sparse node selection.

\begin{center}
\begin{minipage}{\columnwidth}
  \refstepcounter{algorithm}\label{alg:hyper-grad}
  \hrule
  \vspace{1pt}
  \noindent\textbf{Algorithm \thealgorithm} \(\texttt{HyperGrad}\): Approximate Hypergradient Computation
  \vspace{1pt}
  \hrule
  \scriptsize
  \begin{algorithmic}[1]
    \STATE \textbf{Input:} validation loss \(\mathcal{L}_{V}\), training loss \(\mathcal{L}_{T}\)
    \STATE \hspace{1.5em} scores \(\boldsymbol{\lambda}\), model parameters \(\mathbf{w}\)
    \STATE \textbf{Output:} hypergradient \(\partial \mathcal{L}_{V}/\partial \boldsymbol{\lambda}\)

    \STATE \(v \gets \frac{ \partial \mathcal{L}_{V} }{ \partial \mathbf{w} }\)
    \STATE \(p \gets v\)
    \STATE \(H \gets \partial^2\mathcal{L}_{T}/\partial \mathbf{w}\partial \mathbf{w}^{T}\)
    \STATE \(B \gets \partial^2\mathcal{L}_{T}/\partial \mathbf{w}\partial \boldsymbol{\lambda}^{T}\)
    \FOR{\(j = 0\) \textbf{to} \(i\)}
      \STATE \(v \gets v \times (I - H)\)
      \STATE \(p \gets p + v\)
    \ENDFOR
    \STATE \(\frac{\partial \mathcal{L}_{V}}{\partial \boldsymbol{\lambda}} \gets -p \times B\)
    \STATE \textbf{return} \(\partial \mathcal{L}_{V}/\partial \boldsymbol{\lambda}\)
  \end{algorithmic}
  \hrule
\end{minipage}
\end{center}

\begin{center}
\begin{minipage}{\columnwidth}
  \refstepcounter{algorithm}\label{alg:gumbel-softmax}
  \hrule
  \vspace{1pt}
  \noindent\textbf{Algorithm \thealgorithm} \(\texttt{GumbelTopK}\): Gumbel-Softmax Top-\(K\) Sampling with STE
  \vspace{1pt}
  \hrule
  \scriptsize
  \begin{algorithmic}[1]
    \renewcommand{\algorithmicrequire}{\textbf{Input:}}
    \renewcommand{\algorithmicensure}{\textbf{Output:}}
    \REQUIRE scores $\boldsymbol{\lambda} = \{\lambda_1, \dots, \lambda_N\}$, temperature $\tau$, sparsity $K$
    \ENSURE differentiable mask $\mathbf{m} \in [0,1]^N$
    
    \STATE \inlinecomment{Perturbation and Relaxation}
    
    
    \STATE $\mathbf{g} \leftarrow [g_i \sim \operatorname{Gumbel}(0,1)]_{i=1}^{N}$
    \STATE $\mathbf{z} \leftarrow (\boldsymbol{\lambda} + \mathbf{g}) / \tau$
    \STATE $\mathbf{p} \leftarrow \operatorname{Softmax}(\mathbf{z})$

    \STATE \inlinecomment{Discretization via Top-K}
    \STATE $\mathcal{S} \leftarrow\operatorname{arg-top-}K(\{p_i\}_{i=1}^N)$
    \STATE $h_i \leftarrow\mathbf{1}[i \in \mathcal{S}]$ for all $i \in \{1, \dots, N\}$
    
    \STATE \inlinecomment{Gradient Estimation}
    \STATE \(\mathbf{m} \leftarrow K \cdot \mathbf{p} + \operatorname{stop\_gradient}(\mathbf{h} - K \cdot \mathbf{p})\)
    \STATE \textbf{return} \(\mathbf{m}\)
  \end{algorithmic}
  \hrule
\end{minipage}
\end{center}

\section{Analysis}

\subsection{Complexity Analysis}\label{sec:appendix-complexity}

\subsubsection{Time Complexity Analysis}\label{sec:appendix-time-complexity}

The time complexity of Algorithm~\ref{alg:hyper-grad} is primarily determined by the Neumann series approximation used to estimate the inverse Hessian. This procedure involves one initial gradient computation on the validation set, followed by \(i\) iterations of Hessian-vector products (HVPs), and a final projection step. Each HVP can be computed efficiently via automatic differentiation, with the same time complexity as a standard gradient backpropagation pass, denoted \(\mathcal{O}(C_{\text{GNN}})\), without explicitly forming the Hessian matrix. As a result, the overall complexity of the Neumann approximation is linear in the number of approximation steps:
\[
\mathcal{O}(i \cdot C_{\text{GNN}}).
\]

The overall time complexity of Algorithm~\ref{alg:bilevel} is governed by the nested bilevel optimization structure combined with the cost of Graph Neural Network (GNN) propagations. Let \(T\) and \(I\) denote the number of outer and inner iterations, respectively, and let \(C_{\text{GNN}}\) represent the computational cost of a single GNN forward-backward pass, which typically scales as
\(
\mathcal{O}(|\mathcal{E}| D + |\mathcal{V}| D^{2}),
\)
for a graph with \(|\mathcal{V}|\) nodes, \(|\mathcal{E}|\) edges, and node embedding dimension \(D\). In each outer iteration, the algorithm performs \(I\) model update steps and one hypergradient computation involving \(i\) Neumann approximation steps. Hence, the total time complexity is
\[
\mathcal{O}(T \cdot (I + i) \cdot C_{\text{GNN}}).
\]

\bigskip

Additionally, we incorporate the computational cost of the Sparse Querying component. During each outer iteration, Sparse Querying selectively generates explanation texts and re-encodes features only for a sparse subset of nodes \(\mathcal{V}_{expl}\) with cardinality \(K\). The per-iteration cost for this step is
\[
\mathcal{O}\big( K \cdot (C_{\text{LLM}} + C_{\text{LM}}) \big),
\]
where \(C_{\text{LLM}}\) and \(C_{\text{LM}}\) indicate the computational cost of a single LLM inference and LM encoding, respectively.
More precisely, \(C_{\text{LLM}}\) depends on the input and output token lengths of the queried nodes rather than being a length-independent constant. We use \(C_{\text{LLM}}\) as a shorthand for the average per-query cost under our controlled prompting setup. Since we explicitly limit the output length to 50 words, the processed token length is relatively concentrated in practice; in an additional Qwen-27B run, the P10--P90 range of input plus output length is 159--476 tokens. Thus, treating \(C_{\text{LLM}}\) as an average per-query cost is an analytical convenience, not a claim that transformer inference is truly independent of sequence length.

Note that the full initial encoding of all node texts using the fixed LM, which requires \(\mathcal{O}(|\mathcal{V}| \cdot C_{\text{LM}})\) time, is performed once prior to the bilevel optimization and is thus excluded from the iterative complexity.

Consequently, the overall time complexity of Algorithm~\ref{alg:bilevel} including Sparse Querying is
\[
\mathcal{O}\left( T \cdot \left( K \cdot (C_{\text{LLM}} + C_{\text{LM}}) + (I + i) \cdot C_{\text{GNN}} \right) \right).
\]

Because \(K \ll |\mathcal{V}|\) and both LLM and LM are frozen models performing forward-only passes, the Sparse Querying overhead remains well controlled. This design ensures efficient and scalable training of the graph LLM model on large text-attributed graphs.

Since \(C_{\text{LLM}}\) is significantly larger than both \(C_{\text{LM}}\) and \(C_{\text{GNN}}\), and the number of outer iterations \(T\) is fixed and relatively small (set to 3 in our experiments), the overall time complexity of the algorithm~\ref{alg:bilevel} is dominated by the Sparse Querying's LLM inference cost. Therefore, the asymptotic time complexity can be approximated as
\[
\mathcal{O}(K \cdot C_{\text{LLM}}).
\]
This highlights that, despite incorporating sophisticated bilevel optimization and GNN training steps, the computational bottleneck lies in the selective LLM calls for explanation generation on a sparse subset of nodes. Our adaptive sparse querying strategy, which keeps \(K\) substantially smaller than \(|\mathcal{V}|\), thus plays a critical role in enabling the practical efficiency of the proposed framework.

\subsubsection{Space Complexity Analysis}\label{sec:space-complexity}

The Neumann-series approximation in Algorithm~\ref{alg:hyper-grad} does not materialize the Hessian matrix. It computes Hessian-vector products on the fly, so its peak memory remains on the same order as a standard GNN backward pass, denoted by \(\mathcal{O}(M_{\mathrm{GNN}})\), where \(M_{\mathrm{GNN}}\) is the peak memory required by the GNN forward/backward computation.

For the full \method{} framework, besides the GNN memory, we store selector variables \(\boldsymbol{\lambda}\), \(\mathbf{p}\), and \(\mathbf{h} \in \mathbb{R}^{|\mathcal{V}|}\), which require \(\mathcal{O}(|\mathcal{V}|)\) space. The sparse querying module materializes LLM outputs and refreshed LM embeddings only for the selected nodes \(|\mathcal{V}_{\mathrm{expl}}|=K\). Since these nodes are processed sequentially, the additional memory cost is
\[
\mathcal{O}(M_{\mathrm{LLM}} + M_{\mathrm{LM}} + KD).
\]
Therefore, the overall peak space complexity is
\[
\mathcal{O}(M_{\mathrm{GNN}} + |\mathcal{V}| + M_{\mathrm{LLM}} + M_{\mathrm{LM}} + KD).
\]
Since \(|\mathcal{V}| + KD \ll M_{\mathrm{GNN}} + M_{\mathrm{LLM}} + M_{\mathrm{LM}}\) in typical LLM-enhanced TAG settings, this simplifies to \(\mathcal{O}(M_{\mathrm{GNN}} + M_{\mathrm{LLM}} + M_{\mathrm{LM}})\) up to the lightweight selector and selected-embedding storage. In practice, the peak memory is dominated by frozen LM/LLM inference, so \method{} does not introduce a fundamentally new memory bottleneck beyond the language-model components already required by LLM-enhanced TAG learning.

\subsection{Analysis on Ablation Study of Sparse Selection Strategy}\label{sec: ablation-analysis}

To assess the contribution of the learned bilevel selector, we compare it with simpler node-importance rules under the same sparse-query setting. We keep the query budget, backbone, and training settings fixed, and replace only the selection strategy.

\textbf{Selection baselines.}
We evaluate the effectiveness of the bilevel-optimized selection strategy by comparing it against three representative baselines.
\textit{1) Heuristic (Local Dissimilarity)} prioritizes nodes whose features are semantically dissimilar to their neighborhoods. Its score is the cosine distance between a node feature $\mathbf{h}_n$ and the neighborhood mean $\mathbf{h}_n^{\mathcal{N}} = \frac{1}{|\mathcal{N}(n)|} \sum_{m \in \mathcal{N}(n)} \mathbf{h}_m$:
\[
\text{Score}_n^{\text{dissim}} = 1 - \frac{\mathbf{h}_n \cdot \mathbf{h}_n^{\mathcal{N}}}{|\mathbf{h}_n| |\mathbf{h}_n^{\mathcal{N}}|}.
\]
\textit{2) Heuristic (Prediction Entropy)} selects nodes where the base model is most uncertain:
\[
\text{Score}_n^{\text{ent}} = -\frac{\sum_{c=1}^{C} p_{n,c} \log p_{n,c}}{\log C},
\]
where $p_{n,c}$ is the softmax probability for class $c$ and $C$ is the number of classes.
\par\smallskip
\textit{3) LLM-GNN-style Selector} adapts the selector idea from LLM-GNN~\cite{chen2024e} to our semi-supervised feature-enhancement setting, providing a closer LLM-based node selection comparison.
We maintain identical hyperparameters, including budget \(K\), model architecture, and other settings, across all strategies.

\input{tables/ablation_exp.tex}

\textbf{Performance Analysis.}
The results in Table~\ref{tab:ablation-exp} demonstrate that \method{} consistently achieves the highest accuracy across both datasets, outperforming similarity-based, uncertainty-based, and LLM-GNN-style selectors.
Specifically, \textbf{Local Dissimilarity} targets potential outliers via static feature distributions, while \textbf{Prediction Entropy} selects nodes with high model uncertainty.
However, as shown in Table~\ref{tab:ablation-exp}, these baselines provide inconsistent gains; for instance, Prediction Entropy slightly improves performance on Instagram but underperforms on Photo compared to Local Dissimilarity.
The LLM-GNN-style selector is a stronger LLM-aware comparison under our setting, but it still lacks the validation-loss-driven bilevel objective used by \method{}.
In contrast, our bilevel approach \textit{dynamically} identifies nodes that provide the maximum gradient-based utility within the GNN's optimization landscape. By directly optimizing the selection mask for the validation loss, \method{} transcends simple heuristic signals to capture the complex synergy between graph topology and LLM-derived features, leading to more robust performance gains across different graph domains.

\subsection{Analysis on Transferability of Importance Scores}\label{sec: importance-transferability-anslysis}

To assess the transferability of the learned importance scores across different scale GNN architectures, we reuse a selector trained with a shallow proxy model across downstream GNNs with different depths. We keep the sparse query budget fixed and train each backbone using the transferred node scores.

\input{tables/transfer_exp.tex}
We hypothesize that using a shallow proxy mitigates the risks of overfitting and over-smoothing that typically plague bilevel optimization on deep GNNs. Specifically, the 5-layer model's large capacity may lead the selector to overfit the training set by selecting "shortcut" nodes that provide immediate but poorly-generalizable gains. In contrast, the 1-layer model acts as an implicit regularizer, forcing the selector to prioritize nodes whose local structural and textual signals are most informative. Furthermore, deeper GNNs often suffer from over-smoothing, which blurs node distinctiveness and results in fuzzy gradient signals that make it difficult for the selector to identify true "information sources." By utilizing a 1-layer proxy, \method{} preserves individualized node features and avoids the amplification of structural noise, thereby providing cleaner and more transferable importance scores that effectively complement the LLM's knowledge across different GNN depths.

\subsection{Sensitivity Analyses}
\label{app:sensitivity_k}
\paragraph{Selective querying budget \(K\).}
To provide further insights into the trade-off between computational efficiency and predictive performance, we conduct a sensitivity analysis on the hyperparameter $K$, which represents the budget for selective LLM querying. We evaluate \method{} on the Instagram dataset with $K \in \{2, 5, 10, 20, 50\}$. The results are summarized in Table~\ref{tab:k_sensitivity}.
\begin{table}[htbp]
\centering
\caption{Sensitivity analysis of the query budget $K$ on the Instagram dataset. The default setting used in our main experiments is $K=10$ (highlighted in gray).}
\label{tab:k_sensitivity}
\small
\resizebox{\columnwidth}{!}{%
\begin{tabular}{lccc}
\toprule
Budget $K$ & Time (s) $\downarrow$ & Accuracy (\%) $\uparrow$ & Std. Dev. \\
\midrule
2  & 116.94 & 65.63 & 0.39 \\
5  & 122.70 & 65.99 & 0.40 \\
\rowcolor[gray]{0.9} 10 (Default) & 129.04 & 66.22 & 0.55 \\
20 & 144.52 & 66.29 & 0.23 \\
50 & 197.18 & 66.44 & 0.48 \\
\bottomrule
\end{tabular}
}
\end{table}
\paragraph{Observation and Justification for Default $K=10$.}
As observed in Table~\ref{tab:k_sensitivity}, the performance of \method{} is remarkably robust across a wide range of budgets. 
Specifically, even with an extremely restricted budget ($K=2$), \method{} achieves $65.45\%$ accuracy, demonstrating the effectiveness of our adaptive querying strategy in identifying the most critical nodes. 
As $K$ increases from 2 to 10, we observe a steady gain in accuracy (+$0.59\%$) with only a minor increase in computational time ($10.3\%$). However, further increasing $K$ beyond our default value (from 10 to 50) yields diminishing returns: the accuracy only marginally improves by $0.22\%$, while the computational overhead increases significantly by $52.8\%$ (from 129.04s to 197.18s). 
This trend justifies our choice of \textbf{$K=10$} as the default hyperparameter, as it serves as an optimal ``sweet spot'' that balances predictive power and computational efficiency. The stable standard deviation across different values of $K$ further validates the reliability of our bilevel-optimized sparse querying.

\paragraph{Optimization hyperparameters.}
\label{sec:optimization-hyperparameter-sensitivity}

We further evaluate the sensitivity of \method{} to the temperature \(\tau\), the selector learning rate \(\lambda_{lr}\), and the number of Neumann approximation steps \(i\) on Instagram. Tables~\ref{tab:tau-sensitivity}--\ref{tab:neumann-sensitivity} show that both accuracy and runtime remain stable across a broad range of values, suggesting that \method{} is not overly sensitive to these optimization hyperparameters.

\begin{table}[htbp]
\centering
\caption{Sensitivity analysis of temperature \(\tau\).}
\label{tab:tau-sensitivity}
\small
\begin{tabular}{cccc}
\toprule
\(\tau\) & Time (s) & Acc. Mean & Acc. Std \\
\midrule
1 & 130.16 & 66.02 & 0.44 \\
2 & 130.90 & 66.17 & 0.56 \\
\rowcolor[gray]{0.9} 4 & 129.04 & 66.22 & 0.55 \\
8 & 130.87 & 66.11 & 0.51 \\
16 & 133.60 & 66.20 & 0.31 \\
\bottomrule
\end{tabular}
\end{table}

\begin{table}[htbp]
\centering
\caption{Sensitivity analysis of selector learning rate \(\lambda_{lr}\).}
\label{tab:lambda-lr-sensitivity}
\small
\begin{tabular}{cccc}
\toprule
\(\lambda_{lr}\) & Time (s) & Acc. Mean & Acc. Std \\
\midrule
0.0001 & 129.98 & 65.82 & 0.67 \\
0.0005 & 130.70 & 66.22 & 0.27 \\
\rowcolor[gray]{0.9} 0.001 & 129.04 & 66.22 & 0.55 \\
0.005 & 129.13 & 65.99 & 0.25 \\
0.01 & 130.80 & 66.15 & 0.38 \\
\bottomrule
\end{tabular}
\end{table}

\begin{table}[htbp]
\centering
\caption{Sensitivity analysis of Neumann approximation steps \(i\).}
\label{tab:neumann-sensitivity}
\small
\begin{tabular}{cccc}
\toprule
\(i\) & Time (s) & Acc. Mean & Acc. Std \\
\midrule
1 & 127.98 & 65.78 & 0.36 \\
5 & 127.95 & 66.09 & 0.29 \\
\rowcolor[gray]{0.9} 10 & 129.04 & 66.22 & 0.55 \\
20 & 129.06 & 65.85 & 0.71 \\
50 & 129.71 & 66.13 & 0.20 \\
\bottomrule
\end{tabular}
\end{table}

\section{Use of AI Assistants in Manuscript Preparation}
Generative AI tools were used to assist with language polishing, manuscript editing, and LaTeX formatting during the preparation of this paper. The authors reviewed and revised all AI-assisted text, verified the technical claims, equations, citations, and experimental results, and take full responsibility for the final content. No AI tool is listed as an author.

%% file: tables/ablation_exp.tex
\begin{table}[ht]
\centering
\caption{\textbf{Ablation study of the node selection strategy in \method{}.} We compare our bilevel-optimized sparse querying against prediction entropy, a local dissimilarity heuristic, and an LLM-GNN-style selector implemented under our semi-supervised feature-enhancement setting.}
\label{tab:ablation-exp}
\small
\resizebox{\columnwidth}{!}{%
\begin{tabular}{@{}lcc@{}}
\toprule
Selection Strategy       & Instagram Accuracy (\%) $\uparrow$ & Photo Accuracy (\%) $\uparrow$ \\ \midrule
Local Dissimilarity      & 65.89 $\pm$ 0.69         & 85.62 $\pm$ 0.22                \\
Prediction Entropy       & 65.82 $\pm$ 0.47         & 85.54 $\pm$ 0.34                \\
LLM-GNN-style Selector   & 65.98 $\pm$ 0.73         & 85.59 $\pm$ 0.24                \\
Ours (Bilevel-Optimized) & \textbf{66.22 $\pm$ 0.55}         & \textbf{85.75 $\pm$ 0.28}                \\ \bottomrule
\end{tabular}%
}
\end{table}

%% file: tables/transfer_exp.tex
\begin{table}[ht]
\vspace{-2mm}
\centering
\caption{\textbf{Transferability analysis and efficiency of node selection.} We report the Accuracy (\%) on different \textbf{Target} models using node selections optimized on different \textbf{Source} models. The \#Params column indicates the model complexity used during the bilevel optimization.}
\label{tab:transfer-exp}
\small
\resizebox{\columnwidth}{!}{%
\begin{tabular}{lccc}
\toprule
\textbf{Selection Source} & \textbf{\#Params} & \textbf{Target: Small} & \textbf{Target: Big} \\
\midrule
\textbf{Source: Small} & \textbf{2,050} & \textbf{65.89 $\pm$ 0.23} & \textbf{66.37 $\pm$ 0.67} \\
\textbf{Source: Big} & 462,338 & 65.85 $\pm$ 0.44 & 66.20 $\pm$ 0.17 \\
\midrule
\textit{Efficiency Gain} & \textbf{225.5$\times$} & \textit{+0.04} & \textit{+0.17} \\
\bottomrule
\end{tabular}
}
\vspace{-5mm}
\end{table}